\documentclass[conference]{IEEEtran}

\usepackage{graphics}
\usepackage{epstopdf}
\usepackage[pdftex]{graphicx}
\usepackage{stfloats}
\usepackage{array}
\usepackage{float}
\usepackage[cmex10]{amsmath}
\usepackage{amsmath,amsfonts,amssymb}
\usepackage{graphicx,graphics}
\usepackage{enumerate}
\usepackage{color}
\usepackage{verbatim}
\usepackage{amsthm}
\usepackage{multirow}
\usepackage{subcaption}
\usepackage{siunitx}
\usepackage{amsmath,amsthm}
\usepackage{cleveref}
\usepackage[font=footnotesize]{caption}
\usepackage[english]{babel}
\addto\captionsenglish{}

\newtheorem{theorem}{Theorem}

\begin{document}
\title{A Spatial-Spectral Interference Model for Millimeter Wave 5G Applications}

\author{\IEEEauthorblockN{Solmaz Niknam\\ and Balasubramaniam Natarajan}
\IEEEauthorblockA{Department of Electrical and Computer Engineering\\
Kansas State University\\
Manhattan, KS, 66506 USA\\
Email: slmzniknam@ksu.edu; bala@ksu.edu\\}
\and
\IEEEauthorblockN{Hani Mehrpouyan}
\IEEEauthorblockA{Department of Electrical and Computer Engineering\\
Boise State University\\
Boise, ID, 83725 USA\\
Email: hani.mehr@ieee.org}
}

\maketitle

\begin{abstract}
The potential of the millimeter wave (mmWave) band in meeting the ever growing demand for high data rate and capacity in emerging fifth generation (5G) wireless networks is well-established. Since mmWave systems are expected to use highly directional antennas with very focused beams to overcome severe pathloss and shadowing in this band, the nature of signal propagation in mmWave wireless networks may differ from current networks. One factor that is influenced by such propagation characteristics is the interference behavior, which is also impacted by simultaneous use of the unlicensed portion of the spectrum by multiple users. Therefore, considering the propagation characteristics in the mmWave band, we propose a spatial-spectral interference model for 5G mmWave applications, in the presence of Poisson field of blockages and interferers operating in licensed and unlicensed mmWave spectrum. Consequently, the average bit error rate of the network is calculated. Simulation is also carried out to verify the outcomes of the paper.
%Millimeter wave (mmWave) band is a promising candidate for emerging Fifth generation (5G) wireless networks. However, the sensitivity of highly directional mmWave signals to blockages changes the nature of signal propagation in mmWave wireless networks which in turn influences the interference behavior. The interference characteristic can also be affected by the utilization of the open access portion of the spectrum which causes unpredictable interference in spatial and spectral domains. Therefore, we propose a spatial-spectral interference model for 5G networks, in the presence of Poisson field of blockages and interferers operating in licensed and unlicensed mmWave spectrum.
\end{abstract}
\vspace{0.5cm}

\section{Introduction} \label{sec:intro}
One of the key enabling technologies of emerging \emph{fifth generation} (5G) wireless networks is the use of bandwidth in the \emph{millimeter-wave} (mmWave) frequencies, i.e, $30$--$300$ GHz~\cite{Andrew2014what}. However, due to undesirable propagation characteristics of mmWave signals such as severe pathloss, strong gaseous attenuation, low diffraction around objects and large phase noise, this section of spectrum has been underutilized. Having large antenna arrays that coherently direct the beam energy will help overcome the hostile characteristics of mmWave channels~\cite{Rappaport2014Millimeter}. However, utilization of the highly directional beams changes many aspects of the wireless system design. Such directional links (that are susceptible to blockages by obstacles along with the distinct mmWave propagation characteristics), will considerably affect the interference model. In fact, interference in the mmWave band may exhibit an on-off behavior~\cite{Andrew2014what}.

As new applications and standards compete to exploit open access frequencies, coexistence of licensed and unlicensed bands in 5G cellular networks is a critical consideration~\cite{Rappaport2013WillWork}. In addition, unlicensed frequencies provide a viable option for offloading traffic~\cite{Andrew2014what, Niknam2016mmWave}. With such mixed use of licensed and unlicensed bands, interference in the mmWave band may have a more unpredictable behavior that needs to be taken into consideration. In general, users may be randomly distributed in space and could be using a random subset of frequency bands.
There are multiple prior efforts that have focused on modeling the interference behavior. An uplink interference model for small cells of heterogeneous networks has been proposed in~\cite{Dong2016MobilityAware}. However, mmWave specifications in modeling the interference, i.e., the effect of the highly directional links and considerable sensitivity of mmWave beams to blockages are not taken into account. An interference model for wearable mmWave networks considering the effect of blockages has been suggested in ~\cite{Heath2015wearable}. However, the location of the interferers and the blockages are assumed to be deterministic. The authors in~\cite{Heath2012modeling} have suggested an interference model for randomly distributed interferers, using a stochastic geometry based analysis. However, similar to~\cite{Dong2016MobilityAware} and~\cite{Heath2015wearable}, in~\cite{Heath2012modeling}, interferers are considered only in spatial domain. Such a consideration may not be adequate to model the interference in networks operating in both licensed and unlicensed frequency bands, due to the randomness in utilizing the frequencies by terminals that share the same spectrum. Authors in~\cite{Hamdi2009unified} have suggested a spectral-spatial model for interference analysis in networks considering the unlicensed frequency bands. However, the effect of the presence of the blockages in the environment is not taken into account in the model.
%However, similar to~\cite{Heath2012modeling,Dong2016MobilityAware}, in~\cite{Hamdi2009unified} which, as we mentioned previously,
%In addition, the interferers are considered only in the spatial domain which may not be a precise model in case of coexistence of the licensed and unlicensed bands in mmWave 5G networks.
%A circular interference model for heterogeneous networks has been proposed in \cite{Martin2016circular}, where the topology of interferers is considered to be condensed in a power profile along circles around the target receiver, while a low number of principal interferers are taken into account. The latter assumption may limit the application of this model given the highly dense 5G networks. In addition, the locations of the interferers are assumed to be deterministic.
%mmWave specifications in modeling the interference, i.e., the effect of the highly directional links and considerable sensitivity of mmWave beams to blockages are not taken into account.
%blockages are not considered and that may not be a suitable model in mmWave communication networks due to the on-off interference pattern caused by considerable sensitivity of mmWave beams to blockages~\cite{Andrew2014what}.
%Authors in~\cite{Dong2016MobilityAware} have recommended a mobility-aware uplink interference model for heterogeneous networks, in which interference from macro users to the small cell users is calculated within a specific area around the small cells.
%Nevertheless, in ~\cite{Heath2012modeling,Dong2016MobilityAware}, .
In summary, current literature in interference modeling for 5G mmWave networks lacks the consideration for the propagation characteristics in the mmWave band, i.e., severe shadowing caused by highly directional links and the presence of blockages and simultaneous use of both licensed and unlicensed spectrums in the mmWave band.

In this paper, we propose a spatial-spectral model for interference analysis in 5G mmWave short-range wireless technologies while considering the impact of random number of blockages in the environment. Such technologies are a part of standards like IEEE 802.11 ad, wireless HD or short-range operating modes between devices for mmWave 5G cellular systems, where communication links range from 1-10\,m~\cite{Rappaport2014Millimeter}. We derive the closed-form expression of the moment generation function (MGF) of the aggregate interference to a victim receiver, considering blockages in the environment. Then, we use this MGF to derive the bit error rate (BER) expression at the victim receiver and validate it using Monte Carlo simulations of the network.

The remainder of this paper is organized as follows. Section~\ref{sec:sys_model} describes the considered system model. In Section~\ref{sec:int_analsys}, we calculate the closed-form expression of the MGF of the aggregated interference and perform the system evaluation. Section \ref{sec:simulation_results} and \ref{sec:conclusion} present simulation results and the conclusion, respectively.
\begin{figure*}[t]
\centering
\includegraphics[scale=0.48]{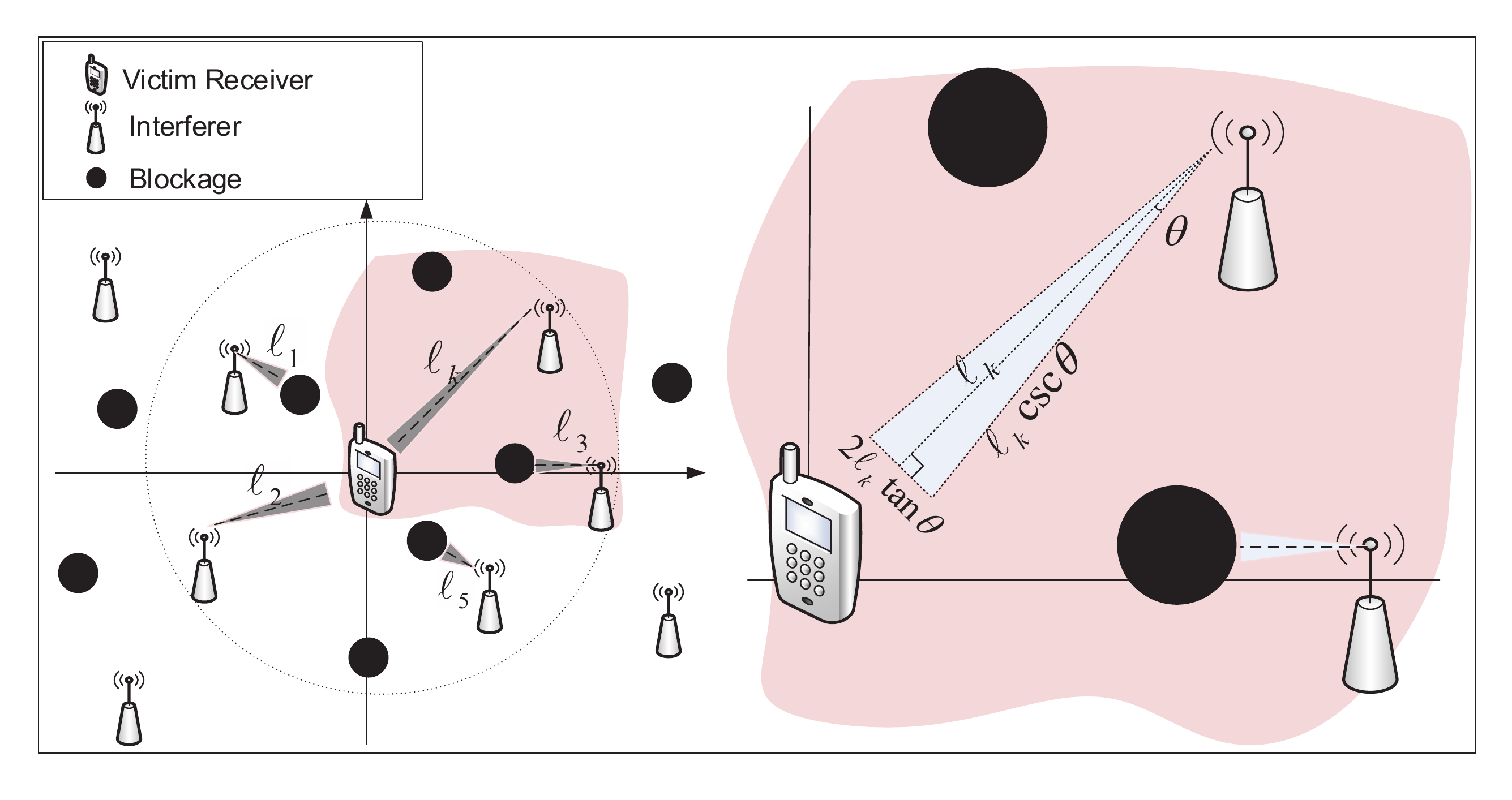}
\caption{The impact of interferers on the victim receiver in the presence of obstacles.}
\label{fig:blockage area}
\end{figure*}

\section{System Model} \label{sec:sys_model}
\vspace{-0.1cm}
We consider a transmitter-receiver pair in the presence of random number of interferers with the receiver at the origin of ${\rm I\!R}^2$ plane communicating with the transmitter over a desired communication link. The number of interferers follows a Poisson point process with parameter $\lambda$ in the space-frequency domain~\cite{Haenggi2009Interference}. We also model the spatial distribution of blockages as a Poisson point process with parameter $\rho$~\cite{Bai2014Blockage}. Considering the large scale signal attenuation, specially in case of mmWave signals that suffer greatly from gaseous attenuation and atmospheric absorbtion, only interference within a limited area around the victim receiver is significant~\cite{Dong2016MobilityAware,Martin2016circular}. A circular area of radius $D$ around the victim receiver is assumed and the number of interferes inside the interfering circle is Poisson distributed with parameter $\lambda \pi D^2$~\cite{Gong2014Interference}. Moreover, in this network, we are primarily concerned with active interferers that are in the line-of-sight (LoS) of the victim receiver. It is important to note that there could be other interferers that do not impact the victim receiver as their signals are blocked by obstacles.
%It is notable that the number of active interferers can vary according to the movement of the target receiver itself and the interferers.
Similar to \cite{Heath2015wearable}, we assume that there is no blockages in the desired communication link. The interferers and their distances to the victim receiver are denoted by $I_k$ and $\ell_k$, for $k=1, 2, ..., U$, respectively. For the $k^\text{th}$ individual interferer, we consider a radiation cone, denoted by $S_k$, where the edges are determined by the beamwidth of the signal. From Fig.~\ref{fig:blockage area}, it is evident that for the $k^\text{th}$ interferer, the radiation cone area, $A_{S_k}$, is given by
\begin{align} \label{cone_area}
A_{S_k}= \frac{{{2\ell_k}\tan \left( { \theta } \right).\,{\ell_k}}}{2} = \ell_k^2\tan \left( \theta  \right).
\end{align}
Similar to~\cite{Heath2015wearable,Bai2014Blockage} and~\cite{Gupta2017blockage}, we assume that the beamwidth of mmWave signals, $2\theta$, is narrow enough that the signal from an interferer is blocked if at least one blockage is presented in the radiation cone of the given interferer. That is, the beamwidth, $2\theta$, is such that the base of the radiation cone of the interferer is smaller than the dimension of the blockage. This considerable sensitivity to blockages results from the high directionality of mmWave signals~\cite{Andrew2014what,Rappaport2014Millimeter}. For instance, measurement results from~\cite{MacCartneyJr2016}, for a transmitter-receiver pair separated by 5\,m, indicate that an average sized body of depth of 0.28\,m causes 30-40\,dB power loss using directional antennas. Therefore, the probability of the $k^\text{th}$ interferer not being blocked, $p_k$, is obtained by
\begin{align} \label{p_not_blocked}
{p_k} = {{\rm e}^{ - \rho A_{S_k}}}={{\rm e}^{ - \rho \ell_k^2\tan \left( \theta  \right)}},
\end{align}
which is consistent with the 3GPP~\cite{3GPPrelease9} and potential indoor 5G 3GPP-like models~\cite{Katsuyuki2016GPP} as well, where the probability of having LoS decreases exponentially as the length of the link increases. Based on the above assumption and system configuration, the received signal at the victim receiver, $R(t)$, is given by
\begin{align} \label{received_signal}
R(t) = {i_0(t)} + \sum\limits_{k = 1}^K {{i_k(t)}}  + n(t),
\end{align}
where $K$ is the number of active interferers, $i_k(t)$ is the signal received from the $k^\text{th}$ interferer, $i_0(t)$ is the desired signal, and $n(t)$ is the additive white Gaussian noise (AWGN) with zero mean and variance $N_0/2$. The received interference signal from the $k^\text{th}$ interferer, can be represented as \cite{goldsmith2005wireless}
\begin{align} \label{received_intf_signal}
{i_k}(t) = \sqrt {{q_k}{h_k}\ell_{k}^{-\alpha}} \,\,{v_k}(t)\,{{\rm e}^{ - j2\pi {f_k}t + {\psi _k}}},
\end{align}
where ${v_k}(t)$ and $q_{k}$ are the baseband equivalent and transmitted power of the $k^\text{th}$ interferer, respectively, $\alpha$ is the pathloss exponent, and $h_k$ is a Gamma distributed random variable that represents the squared fading gains of the Nakagami-$m$ channel model (a generic model that can characterize different fading environments \cite{goldsmith2005wireless}). $f_k$ and $\psi_k$ denote the frequency and phase of the $k^\text{th}$ interferer, respectively, which are assumed to be random~\cite{Haenggi2009Interference}.

\section{Interference Analysis and System Performance} \label{sec:int_analsys}
In this section, the MGF of the accumulated interference is derived and used to quantify the average BER at the victim receiver. Using \eqref{received_signal} and \eqref{received_intf_signal}, the signal to interference and noise ratio (SINR) at the victim receiver can be determined as
\begin{align} \label{SINR}
\text{SINR}=\frac{{{{q_0}{h_0}\ell_{0}^{-\alpha}} }}{{\sum\limits_{k = 1}^K {\mathcal{P}_{I_k}}  + \sigma _n^2}},
\end{align}
where, $\sigma_{n}^{2}$ is the power of the additive noise bandlimited to the signal bandwidth $[-\frac{W}{2}{.}\frac{W}{2}]$. $q_0$, $h_0$, and $\ell_0$ are the transmitted power, the squared channel fading gain, and the distance between desired transmitter and the receiver, respectively. $\mathcal{P}_{I_k}$ is the effective received interference power from the $k^\text{th}$ interferer at the output of the matched filter which is obtained by
\begin{align} \label{effective_interference}
\mathcal{P}_{I_k}={q_k}{h_k}\ell_{k}^{-\alpha}\int_{ - W/2}^{ + W/2} {\Phi (f - {f_k}){{\left| {H(f)} \right|}^2}{\rm{d}}f}.
\end{align}
Here, $H(f)$ is the transfer function of the matched filter on the receiver side, and $\Phi(f)$ is the power spectral density of the baseband equivalent of the interferers' signals. In order to evaluate the performance of the network, we assume that given the distribution of $f_k$, $l_k$, and $\psi_k$, the received interference signal at the output of the matched filter is a complex Gaussian distributed signal. This is a valid assumption as shown in \cite{Chan2005Capacity}. Subsequently, we can relate the BER to SINR as BER=$\frac{1}{2} \text{erfc} \sqrt {c\text{SINR}}$, where $c$ is a constant that depends on modulation used~\cite{goldsmith2005wireless}. In order to find the average BER, we invoke the result in \cite{Hamdi2007Useful}, in which it is shown that the expected value of functions in the form of $g(\frac{x}{y+b})$ can be written as
\begin{align} \label{BER_SINR}
E_x\Bigg[g\left(\frac{x}{{y + b}}\right)\Bigg] = g(0) + \int\limits_0^\infty  {{g_m}({\rm{s}})\,{M_y}(m{\rm{s}}){{\rm e}^{-{\rm{s}}mb}}{\rm{d}}\rm{s}}.
\end{align}
Here, $x$ is a Gamma distributed random variable, ${M_y}(m\rm{s})$ represents the MGF of $y$ in a Nakagami-$m$ fading environment, $b$ is an arbitrary constant, and ${g_m}\left( \rm{s} \right){=}- \frac{{\sqrt c }}{\pi }\frac{{\Gamma (m + \frac{1}{2})}}{{\Gamma (m)}}\frac{{{\kern 1pt} {e^{{\rm{ - }}c{\rm{s}}}}}}{{\sqrt {\rm{s}} }}{}_1{F_1}(1-m;\frac{3}{2};c{\rm{s}})$, where ${}_1{F_1}\left( {a;b;\rm{s}}\right)$ is the confluent hypergeometric function.
In order to utilize \eqref{BER_SINR} to find the average BER, \eqref{SINR} can be rewritten as
\begin{align}
\text{SINR}=\frac{{{h_0}}}{{\frac{1}{{{q_0}\ell _0^{ - \alpha }}}\sum\limits_{k = 1}^K {{\mathcal{P}_{{I_k}}}}  + {{\sigma _n^2} \mathord{\left/
{\vphantom {{\sigma _n^2} {{q_0}\ell _0^{ - \alpha }}}} \right.
\kern-\nulldelimiterspace} {{q_0}\ell _0^{ - \alpha }}}}},
\end{align}
where, $h_0$ is a Gamma distributed random variable and ${{{\sigma _n^2} \mathord{\left/{\vphantom {{\sigma _n^2} {{s_0}\ell _0^{ - \alpha }}}} \right.\kern-\nulldelimiterspace} {{q_0}\ell _0^{ - \alpha }}}}$  is a constant. Therefore, using \eqref{BER_SINR}, the average BER can be written based on the MGF of the accumulated interference as
\begin{align} \label{Average_BER} \notag
&{E_{{h_0}}}\left[ {\text{BER}} \right]= \frac{1}{2} - \frac{\sqrt{c}}{\pi}\frac{\Gamma(m+\frac{1}{2})}{\Gamma(m)} \int_0^\infty  \frac{{\,_1}{F_1}\big( {1 - m;\frac{3}{2};c\rm{s}}\big)}{\sqrt{\rm{s}}} \\
&\hspace{3.5cm}\times{M_I}\left( {m\rm{s}} \right){e^{ - \left( {mb + c} \right)\rm{s}}}{\rm{d}}\rm{s},
\end{align}
where
\begin{align} \label{Interference_MGF}
&{M_I}\left( \rm{s} \right){=} E\left[{{{\rm e}^{ \frac{\rm{s}}{{{q_0}\ell _0^{ - \alpha }}}\sum\limits_{k = 1}^K {{q_k}{h_k}\ell _k^{ - \alpha }\Omega \left( {{f_k}} \right)} }}} \right],\\
&\Omega(f_k)=\int_{ - W/2}^{ + W/2} {\Phi (f - {f_k}){{\left| {H(f)} \right|}^2}{\rm{d}}f}.
\end{align}

Since \eqref{Interference_MGF} is the MGF of sum of a random number of random variables, the distribution of the random variable $K$, i.e., the number of active interferers, is needed.
\newtheorem{lemma}{Lemma}
\begin{lemma}\label{lem:lemma2}
The number of active interferers, $K$, within the circular area of radius $D$ (around the victim receiver) and the signal bandwidth $W$, is a Poisson random variable with parameter $\frac{{\lambda{\pi}{W} \left( {1 - {{\rm e}^{ - {D^2}\rho \tan \theta }}} \right)}}{{\rho \tan \theta }}$.
\end{lemma}
\begin{proof}
Let $K=X_{I_1}+X_{I_2}+...+X_{I_U}$, where $X_{I_k}$ is the indicator that the $k^\text{th}$ interferer is not blocked with probability $p_k$, given by \eqref{p_not_blocked}. In order to make the analysis tractable, we assume that the blockages affect each link independently, i.e., the number of the blockages on different links are independent. The assumption of two links share no common blockages has negligible effect on accuracy~\cite{Bai2014Blockage}. Consequently, $X_{I_k}$ can be modeled as an i.i.d Bernoulli distributed random variable with success probability $p_k$ and $I_k$ is assumed to take on Poisson distribution as presented in section~\ref{sec:sys_model}. Therefore, given the distance $\ell_k$, the probability generating function (PGF) of $X$ is obtained by~\footnote{The subscript of $X_{I_k}$ is dropped for notational simplicity.}
\begin{align}\label{GX_given_l}\notag
{G_{{X|\ell_k}}}\left( {\rm{z}} \right) &= {p_k}{\rm{z}} + (1 - {p_k})\\
&={{{\rm{e}}^{ - \rho \ell_k^2\tan \left( \theta  \right)}}}{\rm{z}} + (1 - {{{\rm{e}}^{ - \rho \ell_k^2\tan \left( \theta  \right)}}}).
\end{align}
In networks with Poisson field of interferers, the probability density function (PDF) of $\ell_k$, i.e., the distance of the $k^\text{th}$ interferer to the victim receiver, is given by \cite{Haenggi2009Interference},
\begin{equation}\label{radial_density}
\mathbb{P}(\ell) = \left\{ \begin{array}{l}
\frac{{2\ell}}{{{D^2}}}\,\,\,\,\,\,\,\,\,0 < \ell < D\\
0\,\,\,\,\,\,\,\,\,\,\,\,\,\text{elsewhere}.
\end{array} \right.
\end{equation}
Based on \eqref{radial_density}, we can average out $\ell$ in \eqref{GX_given_l} leading to
\begin{align}
{G_X}\left( {\rm{z}} \right)=\frac{{\left( {1 - {\rm{z}}} \right)\left( {{{\rm e}^{ - \rho {D^2}\tan \theta }} - 1} \right)}}{{\rho {D^2}\tan \theta }}+1.
\end{align}
Subsequently, the PGF of $K$ is given by
\begin{align}\notag
{G_K}\left( {\rm{z}} \right)&= {\rm E}\Big[ {{{\mathop{\rm z}\nolimits} ^{\sum\limits_{k = 1}^U {{X_{I_k}}} }}} \Big]= \sum\limits_{k \ge 0} {{{\left( {{\rm E}\left[ {{{\mathop{\rm z}\nolimits} ^X}} \right]} \right)}^k}p\left( {U = k} \right)}\\
&= {G_U}\left( {{G_X}({\mathop{\rm z}\nolimits} )} \right)={{\rm e}^{\lambda \pi W \left( {\frac{{1 - {{\rm e}^{ - {D^2}\rho \tan \theta }}}}{{\rho \tan \theta }}} \right)\left( {{\rm{z}} - 1} \right)}},
\end{align}
which is the PGF of a Poisson random variable with parameter $\frac{{\lambda \pi W \left( {1 - {{\rm e}^{ - {D^2}\rho \tan \theta }}} \right)}}{{\rho \tan \theta }}$.
\end{proof}
\begin{theorem}\label{Theo:Theorem1}
The closed-form expression for the MGF of the accumulated interference corresponds to \eqref{Final_Closed_form_MGF}.
\vspace{-0.5cm}
\end{theorem}
\begin{figure*}[t]
\begin{align} \label{Final_Closed_form_MGF}
{M_I}\left( s \right) = {\exp \Bigg\{\lambda \pi W \bigg( {\frac{{1 - {{\rm e}^{ - {D^2}\rho \tan \theta }}}}{{\rho \tan \theta }}} \bigg)\bigg( {{\frac{2}{{{\alpha}W}}\sum\limits_{n = 0}^\infty  {\sum\limits_{j = 0}^\infty  {\prod\limits_{i = 0}^n {\frac{{\,(i + j - \frac{2}{\alpha })\kappa (j)}}{{\Gamma\left( {n} \right)\Gamma\left(j-1\right)}}} } }{\left( {\frac{{{\rm{s}}\,q}}{{{D^\alpha }{q_0}\ell _0^{ - \alpha }}}} \right)^j}\frac{{\Gamma \left( {j + m} \right)}}{{{m^j}\Gamma \left( m \right)}}} - 1} \bigg) \Bigg\} }.\\
\hline
\notag
\end{align}
\vspace{-1.1cm}
\end{figure*}
\begin{proof}
Similar to~\cite{Hamdi2009unified} and \cite{Heath2015wearable}, for simplicity, homogeneous interferers are assumed, i.e., all interferers transmit at the same power. Therefore, given the distribution of $h$, the MGF of the received signal from an arbitrary interferer, ${M_{{I_k}|h}}\left( \rm{s} \right)$, is given by
\begin{align} \label{Intf_MGF_gen} \notag
{M_{{I_k}|h}}\left( \rm{s} \right)&= {\rm{E}}\left[ {{{\rm{e}}^{\frac{{\rm{s}}}{{{q_0}\ell _0^{ - \alpha }}}qh{\ell ^{ - \alpha }}\Omega (f)}}|h} \right] \\ \notag
&= \int_{ - \frac{W}{2}}^{\frac{W}{2}} {\int_0^D {{{\rm{e}}^{\frac{{\rm{s}}}{{{q_0}\ell _0^{ - \alpha }}}qh{\ell ^{ - \alpha }}\Omega (f)}}\frac{{2\ell }}{{{D^2}}}.\frac{1}{W}{\mkern 1mu} {\mkern 1mu} {\rm{d}}\ell {\mkern 1mu} {\rm{d}}f} }\\ \notag
&= \left( {\frac{2}{{{D^2}W\alpha}}} \right){\left( { - \frac{{{\rm{s}}qh}}{{{q_0}\ell _0^{ - \alpha }}}} \right)^{\frac{2}{\alpha }}}\\
&\hspace{1cm}\times \int_{ - \frac{W}{2}}^{\frac{W}{2}} {\Gamma \left( { - \frac{2}{\alpha }, - \frac{{{\rm{s}}qh\Omega \left( f \right)}}{{{D^\alpha }{q_0}\ell _0^{ - \alpha }}}} \right)\Omega {{\left( f \right)}^{\frac{2}{\alpha }}}} {\rm{d}}f.
\end{align}
Here, $f_k$ is a random variable that is uniformly distributed over $[\frac{-W}{2} \frac{W}{2}]$ which is a valid assumption in networks with Poisson field of interferers \cite{Haenggi2009Interference}, and $\Gamma \left( {a,x} \right) \triangleq \int_x^\infty  {{t^{a - 1}}{{\rm e}^{ - t}}dt} $ represents the Incomplete Gamma function. Using the Laguerre polynomials expansion of the Incomplete Gamma function \cite{Abramowitz1954Handbook}, \eqref{Intf_MGF_gen} is simplified to
\begin{align} \label{individual_MGF}
{M_{{I_k}|h}}( {\rm {s}}){=}\frac{2}{{{\alpha}W}}\sum\limits_{n = 0}^\infty  {\sum\limits_{j = 0}^\infty  {\prod\limits_{i = 0}^n {\frac{{(i + j - \frac{2}{\alpha })\,\kappa (j)}}{{\Gamma\left( {n} \right)\Gamma\left(j-1\right)}}} } } {\left( {\frac{{ {\rm{s}}{\mkern 1mu} q{\mkern 1mu} h}}{{{D^\alpha }{q_0}\ell _0^{ - \alpha }}}} \right)^j},
\end{align}
where $\kappa \left( {j} \right) = \int_{-W/2}^{W/2} {\Omega{{(f)}^{j}}{\rm d}f}$. Based on the assumption of general Nakagami-$m$ fading channel, $h$ is a Gamma distributed random variable representing the squared fading gain of the channel. Therefore, the MGF of the interference from the $k^{\text{th}}$ interferer can be expressed as
\begin{align} \label{Final_closed_form_MGF}\notag
&{M_{{I_k}}}\left( \rm s \right)=\frac{2}{{{\alpha}W}}\sum\limits_{n = 0}^\infty  {\sum\limits_{j = 0}^\infty  {\prod\limits_{i = 0}^n {\frac{{\,(i + j - \frac{2}{\alpha })\kappa (j)}}{{\Gamma\left( {n} \right)\Gamma\left(j-1\right)}}} } }\\
&\hspace{4cm}\times{\left( {\frac{{{\rm{s}}\,q}}{{{D^\alpha }{q_0}\ell _0^{ - \alpha }}}} \right)^j}\frac{{\Gamma \left( {j + m} \right)}}{{{m^j}\Gamma \left( m \right)}}.
\end{align}
Assuming i.i.d. interference signals, justified by the fact that the sources of the interference are independent from one another~\cite{Haenggi2009Interference}, we have
\begin{align} \label{Closed_form_MGF}\notag
{M_I}\left( {\mathop{\rm s}\nolimits}  \right) &= {\rm E}\Big[ {{{\rm{e}}^{{\rm{s}}\sum\limits_{k = 1}^K {{I_k}} }}} \Big]= \sum\limits_{k \ge 0} {{{\left( {{\rm E}\left[ {{{\rm{e}}^{{\rm{s}}{I_k}}}} \right]} \right)}^k}p\left( {K = k} \right)}\\
&= {G_K}\left( {{M_{{I_k}}}(\rm{s})} \right)={{\rm e}^{\lambda \pi W \left( {\frac{{1 - {{\rm e}^{ - {D^2}\rho \tan \theta }}}}{{\rho \tan \theta }}} \right)\left( {{M_{{I_k}}}\left( \rm{s} \right) - 1} \right)}}.
\end{align}
Substituting \eqref{Final_closed_form_MGF} in \eqref{Closed_form_MGF}, the closed-form expression for the accumulated interference is determined as in \eqref{Final_Closed_form_MGF}.
\end{proof}
\vspace{0.2cm}
\section{Simulation Results} \label{sec:simulation_results}
In this section, we present numerical results to evaluate the performance of the network based on the proposed interference model and validate the results with Monte Carlo simulation. A network region of an area of 100\,$\text{m}^2$ is considered. The normalized distance between the desired transmitter and receiver is set to 1\,$\text{m}$. We assume the pathloss exponent, $\alpha$, and the shape factor of Nakagami distribution, $m$, are set to 2.5 and 3, respectively. Here, similar to \cite{Heath2015wearable}, the power of all interferers assumed to be the same and set to 0\,$\text{dB}$. The beamwidth of the mmWave signals, i.e. $2\theta$, is set to 20 degrees. An ideal raised cosine (RC) filter is assumed at the receiver's side with roll-off factor of 0. In addition, we consider a raised cosine shaped power spectral density for the interfering signals, as well. It is worth mentioning that the proposed model is not limited to specific power spectral densities of the desired and interferers' signals.\\
In Fig.~\ref{fig:BERvsSNR_lambda}, BER versus SNR is shown for different values of $\lambda$. Here, the density of the number of blockages, $\rho$, is set to $10^{-4}$. As expected, as the density of the number of interferers decreases, the performance of the system improves. For higher SNR values, the level of the error floor depends on the different $\lambda$ values. Fig.~\ref{fig:BERvsSNR_rho} illustrates the performance of the system for different values of $\rho$, considering a fixed density of the number of interferers, i.e., $\lambda{=}10^{-4}$. As it is evident from Fig.~\ref{fig:BERvsSNR_rho}, as the number of blockages increases, the probability of the interferers being blocked increases and consequently the performance of the network improves. Here, having higher blockage density in the network enhances the level of the error floor. This is an important result that indicates mmWave signals sensitivity to blockages can be advantageous in densely deployed networks, where objects and users that serve as obstacles reduces the level of the interference. It is important to note that, in both Fig.~\ref{fig:BERvsSNR_lambda} and \ref{fig:BERvsSNR_rho}, the simulated average BER plots aligns well with the theoretical result from the derived interference model.
Fig.~\ref{fig:blockageVSnot_blocked} shows the BER versus SNR of the victim receiver with and without consideration of the blockages. As it is illustrated, when the presence of the obstacles is considered in the interference model, there is less interference signal introduced to the desired communication link. Unlike traditional wireless environment, the sensitivity of directional mmWave signals to the obstacles in the environment leads to a different interference profile that is effectively captured in the proposed model.
\begin{figure}[t]
\centering
\includegraphics[width=.48\textwidth,height=53mm]{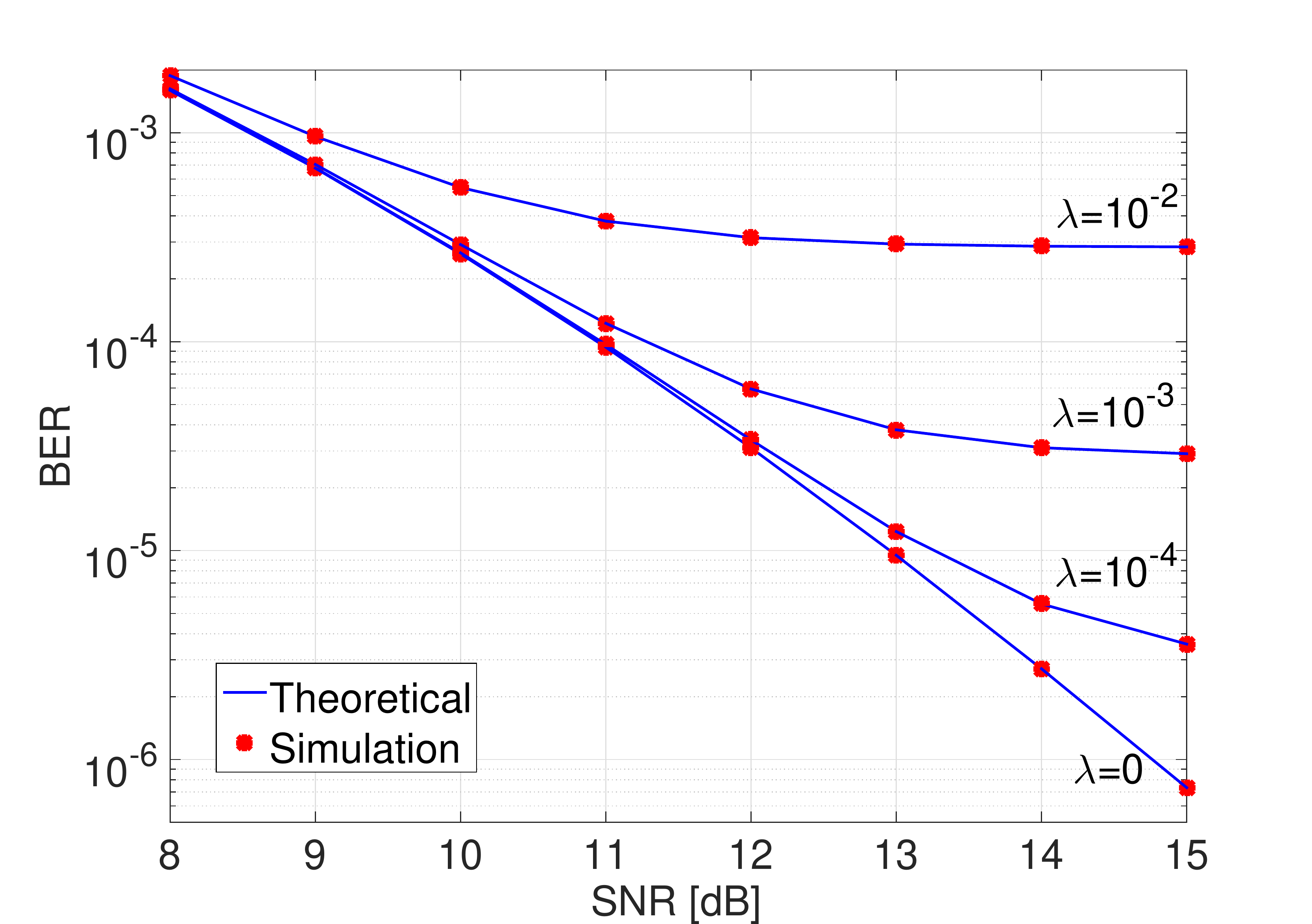}
\caption{Bit error rate versus SNR for different $\lambda$ values, $\rho{=}10^{-4}$.}
\vspace{-0.15cm}
\label{fig:BERvsSNR_lambda}
\end{figure}
\begin{figure}[t]
\centering
\includegraphics[width=.48\textwidth,height=53mm]{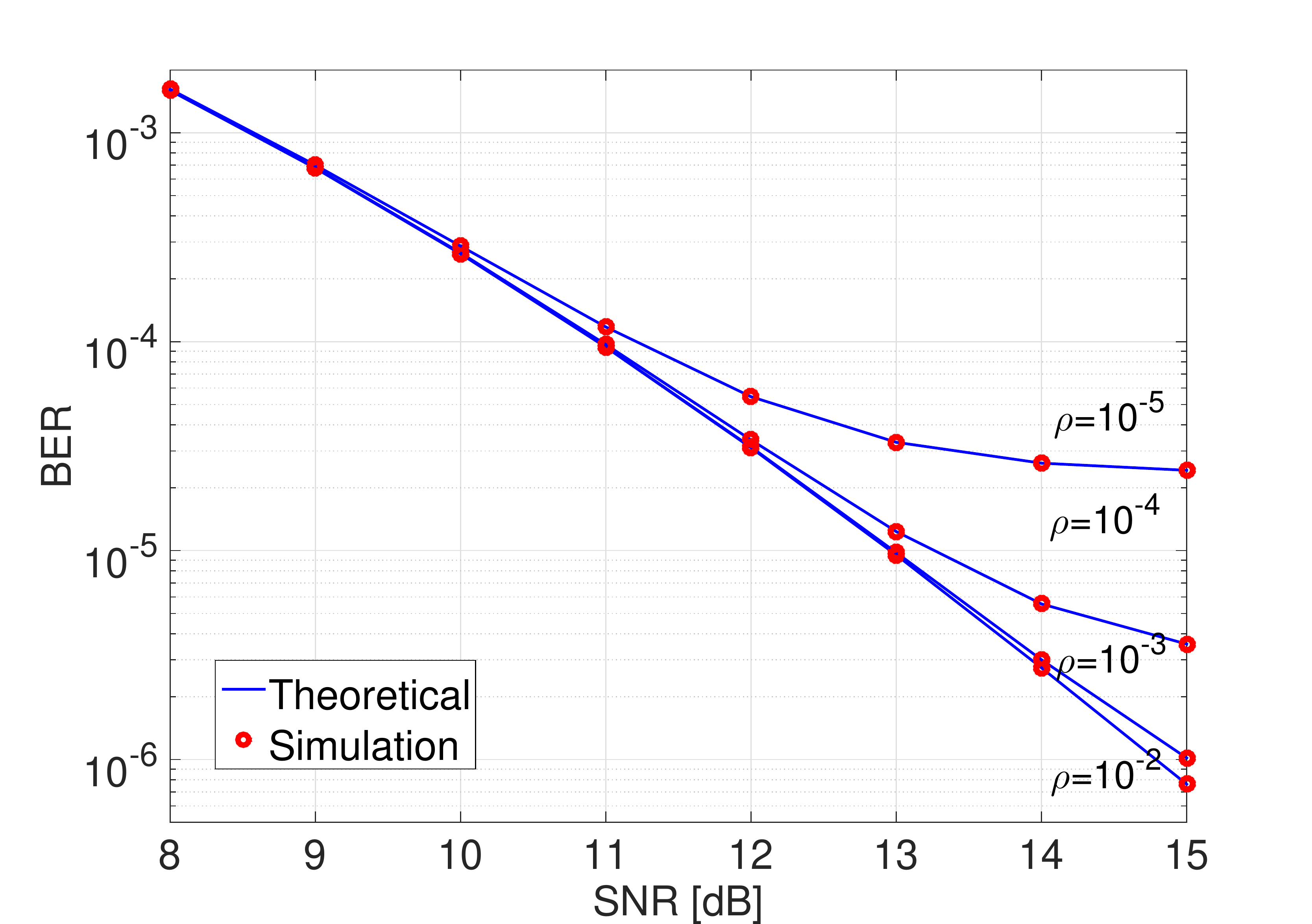}
\caption{Bit error rate versus SNR for different $\rho$ values, $\lambda{=}10^{-4}$. }
\vspace{-0.25cm}
\label{fig:BERvsSNR_rho}
\end{figure}
\begin{figure}[t]
\centering
\includegraphics[width=.48\textwidth,height=53mm]{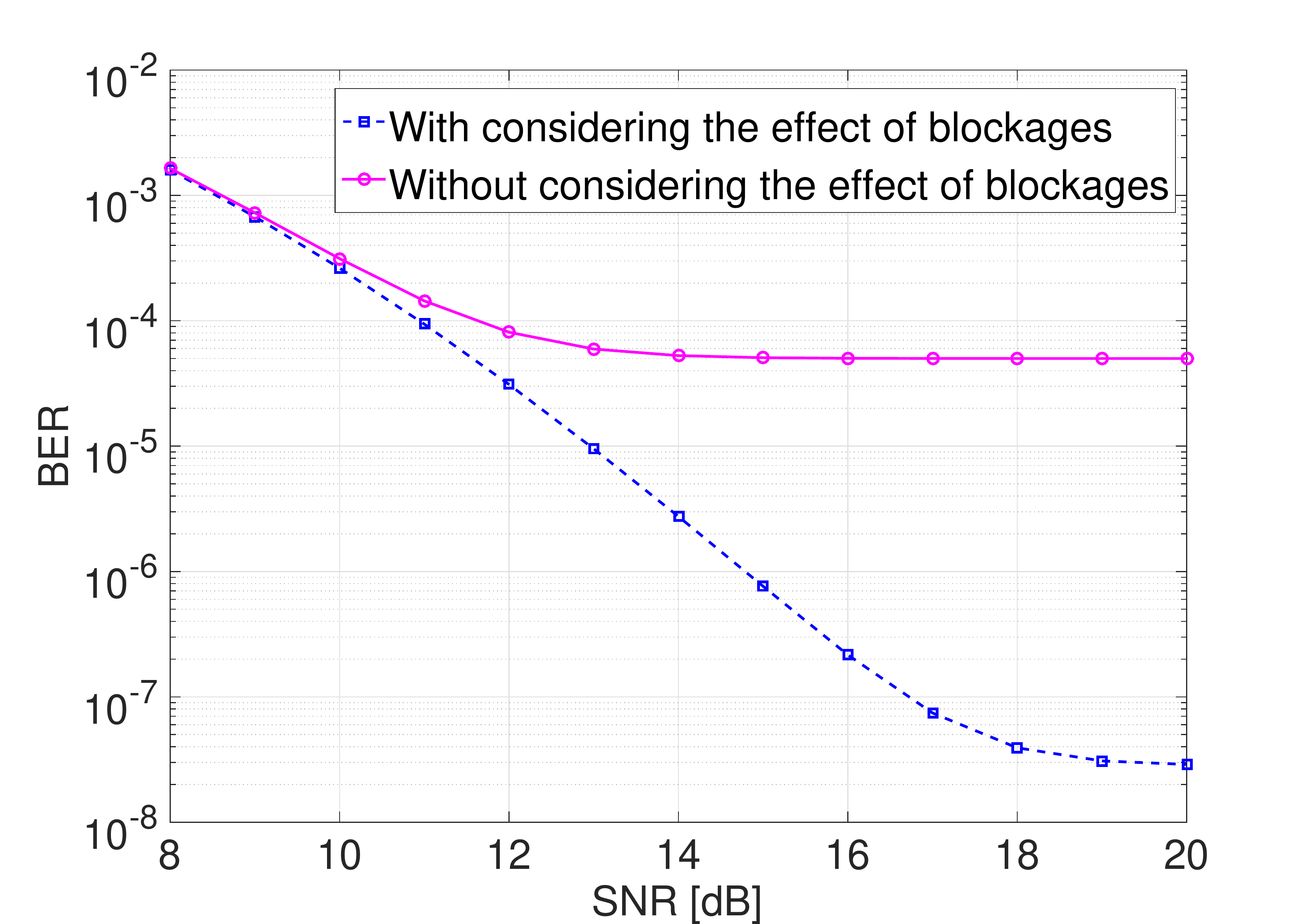}
\caption{Bit error rate versus SNR for $\lambda{=}10^{-4}$ and $\rho{=}10^{-2}$.}
\vspace{-0.15cm}
\label{fig:blockageVSnot_blocked}
\end{figure}

\section{Conclusion and Future Work} \label{sec:conclusion}
In this paper, we analyzed the performance of mmWave communication networks in the presence of Poisson field of interferers and blockages. Due to the use of the unlicensed mmWave frequency band, i.e. the $60$ GHz band, user terminals that share the same spectrum in the network could introduce unpredictable interference to the desired communication links and possible interference could exist in both frequency and space. Considering randomness in the presence of interference in both spectral and spatial domains, we proposed a spatial-spectral model for interference in the network. In the proposed model, MGF of the accumulated interference was derived and based on the closed-form expression of MGF, the average BER at the victim receiver was calculated. In future work, we consider heterogeneous interferers where each interferer transmits at different power level. In addition considering the mobility of the nodes would also be of particular interest.

\bibliographystyle{IEEEtran}

\bibliography{IEEEabrv,GBbibfile}

\end{document}